\documentclass [12pt]{article}
\usepackage{amsmath}
\usepackage{amssymb}
\usepackage{amscd}
\usepackage{amsthm}
\usepackage{amsopn}
\usepackage{xspace}
\usepackage{verbatim}
\usepackage{amsmath}
\usepackage{amsfonts}
\newtheorem{theorem}{Theorem}
\newtheorem{lemma}{Lemma}[section]
\newtheorem{proposition}{Proposition}[section]
\newtheorem{corollary}{Corollary}[section]

\newtheorem{acknowledgment*}{Acknowledgment}

\numberwithin{equation}{section}

\newcommand{\uE}{\overline{E}}
\newcommand{\uP}{\underline{P}}

\newcommand{\oOmega}{\overline{\Omega}}

\newcommand{\be}{\begin{equation}}
\newcommand{\ee}{\end{equation}}
\newcommand{\bd}{\begin{displaymath}}
\newcommand{\ed}{\end{displaymath}}

\newcommand{\R}{\mathbb R}

\newcommand{\C}{\mathbb C}

\begin{document}
\Large
\begin{center}{\bf Concentration of Bloch eigenstates in the presence of gauge at the semi-classical limit  }\end{center}
\normalsize
\begin{center}
Gershon Wolansky\footnote{Department of mathematics, Technion, Haifa 32000, Israel. \  e.mail: gershonw@math.technion.ac.il}\end{center}
\begin{abstract}
We prove a concentration result of a Bloch eigenstate in a periodic channel under a constant gauge. In the semi-classical limit $h\rightarrow 0$ these eigenstates concentrate near a maximizer of the scalar potential of the associated Schr\"{o}dinger operator, provided the constant gauge converges to a critical value from above. This is in contrast with the ground states which concentrate for {\it any} gauge in this limit near a {\it minimizer} of the scalar potential.
\end{abstract}
\section{Introduction}
The effect of a gauge  corresponding to a null magnetic field on the  spectrum of the Schr\"{o}dinger operator is
well known in the case of non simply connected domains [H]. In particular, the spectrum of the Schr\"{o}dinger operator on the circle, manifested by the unit interval with periodic boundary conditions,
\be\label{sch}  L_{h,P}\psi:= -(h d/dx + iP)^2 \psi + V \psi \ ; \ x\in [0,1] \ \ ,
 \psi(0)=\psi(1) \ , \frac{d}{dx}\psi(0)=\frac{d}{dx}\psi(1) \ , \ee
  is affected by the constant gauge $P$. Here $V$ is a smooth, real potential which satisfies the periodic boundary conditions as well.
  \par
  Of particular interest is the effect of this gauge on the spectrum and eigenstates of (\ref{sch}) in the semi-classical limit
  $h\rightarrow 0$.
  The ground state of this operator was extensively studied  (see, e.g, [K]).
It is known that the normalized densities $|\psi|^2$ corresponding to the ground states converge, as $h\rightarrow 0$,  to a Dirac function concentrated at the {\em minimizer} of $V$ (if it is unique), while the ground eigenvalue converges to $\underline{E}:= \min V$. This result is independent of the prescribed gauge $P$. Indeed, the ground state is invariant with respect to the shift $P \mapsto P+ 2\pi h$, so the independence of the asymptotic density of the ground states is plausible. Note, however, that  higher order  features (such as the asymptotics of the spectral gap) are affected by the gauge [H].
\par
On the other hand, the semi classical limits of {\it Bloch states} of the Schr\"{o}dinger operator (\ref{sch}) are sensitive to the gauge already on the leading (macroscopic) order [E].  These are the eigenstates of $L_{h, P}$ which can be presented as $\psi=Ae^{i\theta}$ where the amplitude   $A=|\psi|$ is {\it strictly positive} and the phase  $\theta$  represents a function on the circle (i.e. $\theta(0)=\theta(1)$).
\section{Objectives}
The object of this note is to study some features of the semi-classical  limits of (\ref{sch}) for Bloch states. We focus on the case of {\it super critical Bloch states} corresponding to the eigenvalues above $\uE\equiv\max(V)$. In particular, we show the existence of a sequence of Bloch eigenstates $\psi_h$ whose  amplitude $|\psi_h|$ concentrate, as $h\rightarrow 0$, near a {\it maximizer} of the potential $x_0$ (where $V(x_0)=\uE$). In contrast to the ground states, the existence of such semiclassical limits are strongly conditioned on the gauge $P$.

 Let $V$  be a $1-$periodic potential which is maximized at a {\it unique} point $x_0\in [0, 1)$, such that
 \be\label{cond} \int_0^1(\uE-V)^{-1/2}=\infty\ee
 where $\uE:=\max V$.
Define
$$\uP:= 2^{-1/2}\int_0^1 \sqrt{\uE-V} \ \ . $$
The question we pose is as follows:\vskip .2in\noindent
 {\it  Is there a sequence of normalized
 eigenstates $\psi_h$ of $L_{h,P}$ which, for $h\rightarrow 0$ and $P\searrow \uP$, concentrate near the Dirac function
 $\delta_{x_0}$ while the corresponding eigenvalues converge to $\uE$?}
\vskip .2in

Let the function $E_0=E_0(P)\geq \uE$ defined by inverting
\be\label{invE} \int_0^1 \sqrt{E_0(P)-V} = \sqrt{2}P \ . \ee
The main result we prove is
\begin{theorem}\label{th1}
For any interval $[a,b]\subset ]0, 2\pi[$,  any sequence
 $\{\gamma_n\}\subset [a,b]$ and any $P>\uP$ there exists a sequence $P_n\rightarrow P$  and  a sequence of  normalized Bloch eigenstates    $$L_{h_n, P_n}(\phi_n)=E_0(P_n)\phi_n$$ where $h_n:=\frac{\sqrt{2} P_n}{2n \pi + \gamma_n}$  such that
 \be \label{limdelta}|\phi_n|^2\rightarrow\frac{(E_0(P)-V)^{-1/2}}{\int_0^1(E_0(P)-V)^{-1/2}}:= A^2_0(P) \ . \ee
 holds uniformly on $[0,1]$.
\end{theorem}
The following Corollary satisfies a partial answer to the question above:
{\begin{corollary}
 There exists $P_k\searrow\underline{P}$, $E_k\searrow\uE$,  $h_k\rightarrow 0$  and  normalized Bloch eigenstates
 $$L_{h_k, P_k}(\psi_k)= E_k\psi_k$$ such that $|\psi_{k}|^2\rightharpoonup \delta_{x_0}$ as  distributions.
\end{corollary}
The Corollary follows easily form Theorem~\ref{th1}. Fix $\gamma\in(0, 2\pi)$ and let $\hat{P}_k\searrow \underline{P}$. By  Theorem~\ref{th1} we can find $P_{n,k}\rightarrow\hat{P}_k$ as $n\rightarrow\infty$  and $h_n= \frac{\sqrt{2} \underline{P}}{2n \pi + \gamma}$ so that $\phi_{n,k}$ are normalized Bloch eigenfunctions of $L_{h_n, \hat{P}_{n,k}}$ subjected to the eigenvalue $E(\hat{P}_k)$ and $|\phi_{n,k}|^2\rightarrow A_0^2(\hat{P}_k)$ uniformly as $n\rightarrow\infty$. Then, by Theorem~\ref{th1} it follows
$$ \lim_{k\rightarrow\infty}\lim_{n\rightarrow\infty} |\phi_{n,k}|^2=\lim_{k\rightarrow\infty} A_0^2(\hat{P}_k)=\delta_{x_0}$$
in distributions.
The Corollary then follows by taking a subsequence  $\psi_n=\phi_{n, k(n)}$ where $k(n)\rightarrow\infty$ as $n\rightarrow\infty$ slow enough.
\section{Proof of the main result}
 Each  eigenstate is a critical point of the functional
 $$ F_h(\psi):= \int_0^1 V|\psi|^2 + |h\psi^{'} + iP\psi|^2  \ ; \ \ \int_0^1|\psi|^2 =1$$
 Let us consider the Bloch  states  represented as
 $\psi=Ae^{i\theta/h}$ where $A>0$ and  $\theta$ satisfy the periodic condition on $[0,1]$. Then
 $$ F_h(A, \theta)= \int_0^1 A^2 |\theta^{'}+P|^2 + VA^2 + h^2 |A^{'}|^2$$
Since  $A>0$ on $[0,1]$  there is only one critical point of $F_h$ with respect to $\theta$, namely
$$ \left((\theta^{'} + P)A^2\right)^{'}=0 \ \ \ \Rightarrow \theta^{'} + P = \frac{\lambda}{A^2}$$
for some $\lambda\in \R$. Since $\theta$ satisfies the periodic boundary condition  $\theta(0)=\theta(1)$  it follows
$$ P=\lambda\int_0^1 A^{-2}  \Rightarrow A^2(\theta^{'}+P)^2 = \lambda^2/A^2= P^2 A^{-2} \left(\int_0^1 A^{-2}\right)^2 \ , $$

so $\overline{F}_h(A):= \min_{\theta}F_h(A, \theta)$ takes the value
$$ \overline{F}_h(A)=P^2\left(\int_0^1A^{-2}\right)^{-1} + \int_0^1 VA^2 + h^2 |A^{'}|^2  \ . $$

A critical point of $\overline{F}_h$, then, satisfies
\be\label{el} h^2 A^{''} = (V-E_h(P))A + \frac{ 2P^2 A^{-3}}{\left(\int_0^1 A^{-2}\right)^2} \ \ , \ \ \ A(0)=A(1), \ \ \ A^{'}(0)=A^{'}(1) \ . \ee
 Here $E_h(P)$ is the Lagrange multiplier due to the constraint $\int_0^1A^2=1$ and corresponds to an eigenvalue of the Shrodinger operator.

So, let us set $h=0$ in (\ref{el}) to obtain
$$ A_0^2=\frac{\sqrt{2}P}{\int_0^1 A_0^{-2}}(E_0(P) -V)^{-1/2} \ . $$
This can be solved only if $E_0(P)$ is compatible with (\ref{invE})
and, in particular, only for  $P>\uP:= 2^{-1/2}\int_0^1 \sqrt{\uE-V}$. In that case  we get
$$ A^2_0(P)=\frac{(E_0(P)-V)^{-1/2}}{\int_0^1(E_0(P)-V)^{-1/2}}  \ . $$

Setting
$$ \lambda:= \left( \int_0^1(E_0(P)-V)^{-1/2}\right)^{-2} \ \ \ , \ \ \ E(\lambda):= E_0(P) \ \ , $$
Theorem~\ref{th1} is obtained from
\begin{proposition}\label{prop}
Let
 \be\label{omegalam}\oOmega_\lambda:=   2 \int_0^1\sqrt{ E(\lambda)-V}\ee
and \be\label{omegalam1} A_0(\lambda):= \frac{\lambda^{1/4}}{(E(\lambda)-V)^{1/4}}  \ . \ee
Then, for each each $\lambda>\left( \int_0^1(\uE-V)^{-1/2}\right)^{-2}$ and $\gamma\in (0, 2\pi)$  there exists $N(\gamma)>0$ and a
 solution $A_h$ of
$$ h^2 A_{h}^{''} = -(E(\lambda)-V)A_{h} + \lambda A_{h}^{-3}$$
which satisfies the periodic boundary condition on $[0,1]$
provided $h=\frac{\oOmega_{\lambda}}{2n \pi + \gamma}$ and any integer $n> N(\gamma)$. Moreover, $A_h\rightarrow A_0$ as $n\rightarrow\infty$.
\end{proposition}
\begin{proof}
Let $A_h=A_0+\eta$. Then $\eta$ satisfies
$$ h^2\eta^{''}= -\Omega^2_\lambda(x)\eta + Q(x, \eta, h) \eta  + h^2A_0^{''}$$
where
\be\label{omegadef}\Omega^2_\lambda(x)= E(\lambda)-V +3\lambda A_0^{-4}(\lambda)\equiv 4(E(\lambda)-V)\ee by (\ref{omegalam}, \ref{omegalam1}) and $Q(x,0, h)\equiv 0$. We now scale $x\rightarrow x/h$ $d/dx()\rightarrow  \dot{()}$ and  $\eta\rightarrow h \eta$ to obtain
$$ \ddot{\eta}= -\Omega^2_h(h x)\eta + h \left( \hat{Q}_h(h x, \eta)\eta -A_0^{''}(h x)\right) $$
where
\be\label{qdef}\hat{Q}_h(\eta,x):= h^{-1}Q(h\eta, x, h)=: q(x) \eta + h q_h(x, \eta) \ . \ee

Let us now set
$$ \eta=R\cos\Theta \ , \ \ \ \dot{\eta}= \Omega_\lambda(h x)R\sin\Theta \ . $$
so
\be\label{1} \dot{R} \cos\Theta - \dot{\Theta}R\sin\Theta= \Omega_\lambda(h x) R\sin\Theta\ee
\be\label{2}\left(\dot{R} \sin\Theta + \dot{\Theta}R \cos\Theta\right)\Omega_\lambda(h x)= -\Omega^2_\lambda R\cos\Theta+ h H_h(h x, R, \Theta) \ee
where
 \be\label{hdef}H_h=  \hat{Q}_h(h x, R\sin\Theta)R\sin\Theta -A_0^{''}(h x)-\Omega^{'}_\lambda(h x)R\sin\Theta\ee
Multiply (\ref{1}) by $\cos\Theta$, (\ref{2}) by $\Omega_\lambda^{-1}\sin\Theta$ and sum to obtain
\be\label{sys1} \dot{R}=h \frac{\sin\Theta}{\Omega_\lambda} H_h(h x, R,\Theta)\ee
Likewise, multiply (\ref{1}) by $-\sin\Theta/R$, (\ref{2}) by $\cos\Theta/(R\Omega_\lambda)$ and sum to obtain
\be\label{sys2} \dot{\Theta}= -\Omega_\lambda(h x)+ \frac{h\cos\Theta}{R\Omega_\lambda}H_h(h x , R, \Theta) \ . \ee
In complex notation,  $Z=Re^{i\Theta}$ and (\ref{sys1}, \ref{sys2}) takes the form of single equation in the complex 
plane $\C$:
\be\label{sys3}
\dot{Z}= -i\Omega_\lambda(h x) Z + ih \frac{H_h(h x,Z, Z^c)}{\Omega_\lambda(h x)}\ee
where $Z^c:= Re^{-i\Theta}$.
 Let $\Psi_h:\C\rightarrow\C$ be the map obtained from the solutions of (\ref{sys3}) at time $h^{-1}$, that is:
$$ \Psi_h\left(Z(0)\right):=Z(h^{-1}) \ . $$
We need to prove that $\Psi_h$ has a fixed point for $h$ sufficiently small, under the stated conditions. 

Apply the transformation
\be\label{ztrans} Z \mapsto \tilde{Z}:= Z-h\frac{\cos\theta}{R\Omega_\lambda}\ee
and substitute in (\ref{sys3}) to obtain
\be\label{sys4} \dot{\tilde{Z}}= -i\Omega_\lambda(h x) \tilde{Z}+ O(h^2) \ee
Integrating (\ref{sys4}) on the interval $[0, h^{-1}]$ to obtain
$$ \tilde{Z}\left( h^{-1}\right) = \tilde{Z}(0)e^{-i\overline{\Omega}_\lambda/h} + O(h)$$
where $\oOmega_\lambda=\int_0^1\Omega_\lambda(x) dx$ as defined in (\ref{omegalam}).

Note that the transformation (\ref{ztrans}) is singular at $R=0$, so the term $O(h^2)$  is also singular at $R=0$. However, the conjugancy  between (\ref{sys3}) and (\ref{sys4}) is still valid away from $R=0$. In particular we obtain
\begin{lemma}
 The limit
 $$ \lim_{h\rightarrow 0} e^{i\oOmega_\lambda/h}\Psi_h(Z)=Z$$
 holds uniformly on  $|Z|=r$ for each $r>0$. In particular, if $h_n=\frac{\oOmega_{\lambda}}{2n \pi + \gamma}$  then
 \be\label{zlim} \lim_{n\rightarrow \infty} \Psi_{h_n}(Z)=e^{-i\gamma}Z\ee
 uniformly on $|Z|=r$.
 \end{lemma}
 We complete the proof by utilizing the following fixed point theorem which follows from a  "soft" topological argument
 \begin{lemma}
 Assume $\Psi: \{|Z|\leq r\} \rightarrow \C$ is a  of continuous function.   Given $\gamma\in(0, 2\pi)$, if $|\Psi(Z)-e^{-i\gamma} Z|$ is sufficiently small for $|Z|=r$
 then $\Psi$ has a fixed point $z_0$ in $\{|Z|<r\}$.
 \end{lemma}
 \begin{proof}
 Assume the contrary. Then the function
 $$ \Phi(z)= \frac{\Psi(Z)-Z}{|\Psi(Z)-Z|}: \{|Z|\leq r\} \rightarrow \{|Z|=1\}$$
 is continuous. By assumption
 $|\Phi(Z)-e^{i\alpha}Z|$ is uniformly small on $|Z|=r$ where $e^{i\alpha}= \frac{1-e^{-i\gamma}}{|1-e^{-i\gamma}|}$. In particular, the winding number of the mapping $\Psi$ restricted to the circle $\{|Z|=r\}$ equals one, so the topological degree of $\Phi$ equals one as well. By degree theory [OCQ] $\Phi=0$ must have a root in the disk $\{|Z|\leq r\}$, which is clearly impossible.
 \end{proof}
 The proof of Proposition~\ref{prop} (and Theorem~\ref{th1}) now  follows since $r>0$ can be chosen arbitrary small if $h$ is small enough. In particular, the amplitude of $\eta=A_h-A_0$ is arbitrary small for $h$ small enough.
\end{proof}
\begin{center}{\bf References}\end{center}
\begin{description}
 \item{[E]} Evans L.C: {\it Towards a Quantum Analog of Weak KAM Theory }  , Communications in Mathematical Physics, Volume 244, Issue 2, (2004) 311-334
\item{[H]} Helffer, B: {\it Semi-Classical Analysis for the Schr\"{o}dinger Operator and Applications}, Lec. Notes in Math., 1336, (1988)
\item{[OJQ]}   O'Regan, D ,  Je Cho, Y and  Yu-Qing, C:  {\it Topological Degree Theory and Applications}, Series in Mathematical Analysis and Applications, (2006)
 \item {[K]} Kuchment P: {\it Floquet Theory for Partial Differential Equations}, Operator Theory, Adv. and Appl., Vol. 60, Birkhauser Verlag 1993
\end{description}

\end{document}